\documentclass[11pt]{article}
\usepackage[T1]{fontenc}

\usepackage{tgheros}
\usepackage{amsmath,amsfonts,amsthm,mathtools}
\usepackage{mathtools}
\usepackage{fullpage}
\usepackage{thmtools,thm-restate}
\usepackage[linesnumbered,noend,ruled,noline]{algorithm2e}
\usepackage{setspace}
\usepackage{tikz}

\SetAlFnt{\small}
\SetAlCapFnt{\small}
\SetAlCapNameFnt{\small}
\SetAlCapHSkip{0pt}
\IncMargin{-\parindent}

\usepackage{graphicx}
\usepackage{amsthm}
\usepackage{hyperref}
\usepackage{cleveref}
\usepackage[numbers]{natbib}
\setcitestyle{acmnumeric}
\DontPrintSemicolon

\SetCommentSty{mycommfont}
\SetArgSty{textnormal}
\usepackage{soul}
\usepackage[title]{appendix}
\usepackage{apxproof}
\usepackage{enumitem}
\usepackage{verbatim}
\usepackage{authblk}

\usepackage[margin=1in]{geometry} 

\newtheorem{theorem}{Theorem}
\newtheorem{lemma}[theorem]{Lemma}
\newtheorem*{lemma*}{Lemma}
\newtheorem*{theorem*}{Theorem}

\newtheorem{claim}[theorem]{Claim}
\newtheorem{remark}[theorem]{Remark}

\newcommand{\vcg}[2]{VCG({#1},{#2})}
\newcommand{\naturals}{\mathbb{N}}
\newcommand{\E}{\mbox{\bf E}}
\title{Multi-Parameter Mechanisms for Consumer Surplus Maximization\thanks{This material is based upon work supported by the National Science Foundation under Grant No. DMS-1928930 and by the Alfred P. Sloan Foundation under grant G-2021-16778, while all three authors were in residence at the Simons Laufer Mathematical Sciences Institute (formerly MSRI) in Berkeley, California, during the Fall 2023 semester. T. Ezra is also supported by the Harvard University Center of Mathematical Sciences and Applications.}}
\date{}
\begin{document}
\author[a]{Tomer Ezra\thanks{tomer@cmsa.fas.harvard.edu}}
\author[b]{Daniel Schoepflin\thanks{ds2196@dimacs.rutgers.edu}}
\author[c]{Ariel Shaulker\thanks{ariel.shaulker@weizmann.ac.il}}
\affil[a]{Harvard University}
\affil[b]{Rutgers University -- DIMACS}
\affil[c]{Weizmann Institute of Science}
\date{} 
\setcounter{Maxaffil}{0}
\renewcommand\Affilfont{\itshape\small}

\maketitle
\begin{abstract}
    We consider the problem of designing auctions that maximize \emph{consumer surplus} (i.e., the social welfare minus the payments charged to the buyers). In the consumer surplus maximization problem, a seller with a set of goods faces a set of strategic buyers with private values, each of whom aims to maximize their own individual utility.  The seller, in contrast, aims to allocate the goods in a way that maximizes the \emph{total buyer utility}.  The seller must then elicit the values of the buyers in order to decide what goods to award each buyer. The canonical approach in mechanism design to ensure truthful reporting of the private information is to find appropriate prices to charge each buyer in order to align their objective with the objective of the seller.  Indeed, there are many celebrated results to this end when the seller's objective is \emph{welfare maximization} 
\cite{vickrey1961,clarke1971,groves1973}
or \emph{revenue maximization} 
\cite{myerson1981optimal}.  However, in the case of consumer surplus maximization the picture is less clear -- using high payments to ensure the highest value bidders are served necessarily decreases their surplus utility, but using low payments may lead the seller into serving lower value bidders.  

Our main result in this paper is a framework for designing mechanisms that maximize consumer surplus.  We instantiate our framework in a variety of canonical \emph{multi-parameter} auction settings (i.e., unit-demand bidders with heterogeneous items, multi-unit auctions, and auctions with divisible goods) and use it to design auctions achieving consumer surplus with tight approximation guarantees against the total social welfare.  Along the way, we resolve an open question posed by \citet{hartline2008optimal}
for the two bidder single item setting.

\end{abstract}
\newcommand{\sw}[2]{SW(#1,#2)}
\newcommand{\sws}[1]{SW(#1)}
\newcommand{\al}[2]{A(#1,#2)}
\newcommand{\ali}[3]{A_{#1}(#2,#3)}
\newcommand{\allocs}{\Delta}

\newcommand{\items}{\mathcal{I}}
\newcommand{\agents}{\mathcal{N}}
\newcommand{\bidders}{\agents}

\newcommand{\mech}{\mathcal{M}}
\newcommand{\opt}{\texttt{OPT}}
\newcommand{\vali}[1][i]{v_{#1}}
\newcommand{\payi}[1][i]{p_{#1}}
\newcommand{\pricej}[1][j]{p_{#1}}
\newcommand{\utili}[1][i]{u_{#1}}

\newcommand{\reals}{\mathbb{R}}

\newcommand{\support}[1]{\text{supp}(#1)}

\newcommand{\assumption}{no-superstar-item }

\makeatletter
\newcommand{\headline}[2]{%
  \vspace{0.08in}%
  \noindent \textbf{#1} #2%
  \vspace{0.08in}%
  \par\nobreak
  \@afterindentfalse
  \@afterheading
}
\makeatother

\thispagestyle{empty}
\newpage
\pagenumbering{arabic}

\section{Introduction}

The field of mechanism design in economics aims to design tools to direct the incentives of a group of self-interested agents toward an objective specified by a central planning mechanism designer.  The self-interested agents hold some private, objective-relevant information and the goal in mechanism design is, generally, to align the individual objectives of the self-interested agents with the central objective of the designer.  As a concrete example, in the canonical ``\emph{single-item auction} problem'' an
\emph{auctioneer} needs to allocate a single, indivisible good to one of several 
\emph{bidders}, each of whom has a \emph{private value} for receiving the good.  The auctioneer may wish to allocate the item to the bidder with the highest value, thereby maximizing the \emph{social welfare} (i.e., the total value) produced by the good, but to do so must elicit the private value information from each of the bidders.  Each bidder, on the other hand, can strategically (mis)report her private value to the auctioneer with the goal of maximizing her own \emph{utility} (i.e., her value for the outcome minus any costs she suffers).  There is, thus, a tension between the objective of the auctioneer and the objectives of the bidders and in order to achieve her goal, the auctioneer must \emph{incentivize} the bidders to report their value information accurately.

The use of monetary transfers is arguably the most ubiquitous approach in mechanism design for incentivizing truthful reporting of private information.  In our single-item auction setting, for example, the well-known second-price auction, which awards the item to the bidder with the highest reported value and charges this bidder a price equal to the second highest reported value, achieves exactly the goal of our auctioneer. Namely, all bidders are incentivized to report their true values to the auctioneer and the social welfare is maximized.  By using transfers of money from the agents to the designer, one is able to design mechanisms with similar properties for a broad class of allocation problems.  Indeed, when an auctioneer aims to maximize social welfare the celebrated Vickrey-Clarke-Groves (VCG) mechanism \cite{vickrey1961,clarke1971,groves1973} works  
 in settings with multiple items even if the bidders have \emph{arbitrarily complex} valuations over subsets of items.
By carefully constructing appropriately large payments from the bidders to the auctioneer, the VCG mechanism ensures that bidders truthfully report all their private information and exactly maximizes the social welfare.

On the other hand, there are many scenarios in which a mechanism designer may wish to limit or reduce the payments she collects from the bidders.  First, in some settings monetary payments may be undesirable or even illegal and, as such, alternative forms of ``payment'' which offer no real revenue to the auctioneer are employed.  For instance, so-called ``ordeal mechanisms'' for allocating welfare benefits \cite{alatas2013ordeal,nichols1982targeting,zeckhauser2021strategic,sylvia2022ordeal} may utilize long applications or waiting time in a queue.  In these mechanisms time and effort stand in for money as a form of payment.  Second, in ``proof of work'' environments 
\cite{dwork1992pricing,jakobsson1999proofs},
such as the Bitcoin blockchain, agents all want access to some shared resource, e.g., the rights to publish the next block, and effort in finding a solution to a difficult computational problem is used to decide which agent(s) are awarded access.  This solution to the computational problem has no value and effort is merely used as ``payment'' to coordinate the agents.  
Third, in many allocation settings with computing resources the ``payments'' come in the form of reduced service quality.  For instance, packets may be delayed by servers \cite{cole2003much} or some resources may go unused \cite{cole2013mechanism}.  

In each of the aforementioned settings the payments have no direct benefits outside of ensuring good incentives and, as such, the payments from the bidders are effectively ``burned''.  
Therefore, a natural goal is to find mechanisms that: (i) still ensure truthful reporting, (ii) find outcomes of high welfare, and (iii) keep the total payments low.
Notably, however, although the VCG mechanism is truthful and maximizes the social welfare, satisfying the first two goals, the payments it specifies can be quite high and may leave, essentially, no residual \emph{consumer surplus} utility for the bidders themselves.  For instance, in the second-price auction when the second highest value is nearly the same as the first highest, the consumer surplus is nearly $0$.

Aiming to address this issue, \citet{hartline2008optimal} introduced the study of ``money-burning mechanisms'' --  mechanisms which may charge payments from the bidders but seek to maximize the residual consumer surplus.  Focusing on the single-item setting, they identify optimal Bayesian mechanisms (i.e., optimal mechanisms when bidders draw their values from distributions known to the auctioneer) for consumer surplus.  Then, they extend these results into worst-case prior-free settings by giving a truthful mechanism that always obtains consumer surplus within a $O(\log(n))$-factor of the total social welfare where $n$ is the number of bidders and demonstrate that this is the best possible.  Finally, they expand their results from single-item settings to single-parameter $m$-unit settings where an auctioneer has $m$ identical items and each bidder aims to receive one of these $m$ items.  

From the groundbreaking work of \citet{hartline2008optimal}, many natural questions arise.  In particular, they call out two specific questions, namely, (i) ``quantify[ing] the power of money burning'' in ``settings beyond i.i.d. distributions and $k$-unit [unit-demand] auctions''; and (ii) finding improved ``upper and lower bounds for ... money-burning mechanisms with a small number of agents'', where the ``best-possible approximation ratio ... even in the two agent, one unit special case'' remained unknown.  In this work, we tackle both these questions and (partially) resolve them.  We study three well-established \emph{multi-parameter} auction settings which all generalize the single-item setting studied in \cite{hartline2008optimal} and give mechanisms achieving asymptotically tight consumer surplus guarantees.  We then narrow our attention to the two-bidder single-item setting and identify exactly the optimal mechanisms therein.

\subsection{Our Contribution and Techniques}\label{sec:contribution}

Maximizing residual consumer surplus poses challenges that do not arise in welfare maximization, where the VCG mechanism can be applied directly. The difficulty lies in balancing prices: setting them too high reduces bidders’ utility, while setting them too low risks allocating items inefficiently.

To address this challenge, we introduce the \textit{VCG with copies} framework (Section~\ref{sec:vcg-copies}). The key idea behind this framework is to run the VCG mechanism on a modified instance instead of the original one. In this modified instance, multiple copies of each item are introduced, and agents' valuations are adjusted so that they do not strictly prefer allocations that were unavailable in the original setting. This approach effectively results in more balanced prices that are lower than those produced by VCG in the original instance. 

The framework consists of four steps. First, the number of copies, or a distribution over possible numbers of copies, is determined. Second, the VCG mechanism is run over the modified instance. Third, a rounding scheme takes the VCG allocation for the modified instance and turns it into an allocation for the original instance where we only have a single copy of each item. The final step adjusts the VCG payments proportionally to the rounded allocation. The third step of using a rounding scheme is nontrivial, as it must ensure that the final allocation remains close to the allocation in the modified instance while maintaining truthfulness.

We instantiate this framework across three multi-parameter auction settings, achieving tight approximation guarantees against the social welfare \footnote{All our lower bounds follow from the bound in \cite{hartline2008optimal}.}. The resulting mechanisms are prior-free, while the matching lower bounds hold even for Bayesian mechanisms\footnote{These lower bounds hold for instances where agent values are independently drawn from a common distribution, and mechanisms need only be Bayesian incentive compatible.}.


First, we analyze multi-unit auctions where bidders have submodular valuation functions. We show that applying our framework in this setting yields a tight logarithmic approximation to the optimal welfare (Theorem~\ref{thm:mu}).


Next, we consider unit-demand bidders. For adversarial heterogeneous items, we show that an $O(\log n)$ approximation to the welfare is achievable (Theorem~\ref{thm:ud}), and this bound is tight due to \cite{hartline2008optimal}.
In Bayesian settings, we refine our analysis of the mechanism's approximation guarantee. Specifically, if the distributions of the agents' valuations satisfy the \assumption assumption that assumes the probability of any item being the favorite item of some agent is at most $\frac{c}{m}$,
we obtain an improved approximation of $O(\log(1+\frac{n\cdot c}{m}))$ (Theorem~\ref{thm:udB}). This generalizes the results of both \cite{hartline2008optimal} and \cite{goldner2024simple}. In particular, \citet{hartline2008optimal} showed that for identical items there exists a mechanism whose residual surplus provides an $O\left(\log \left(1+\frac{n}{m}\right)\right)$ approximation to the welfare, and this bound is asymptotically tight. Later, \citet{goldner2024simple} considered the i.i.d. case, where they presented a simple mechanism obtaining the same approximation (of the welfare) of $O\left(\log \left(1+\frac{n}{m}\right)\right)$, which is tight. Notably, both are special cases of our \assumption assumption with $c = 1$. At the other extreme, when $c = m$ (corresponding to the adversarial heterogeneous items), we recover our bound for heterogeneous items. To devise our mechanism, we take a different approach from \citep{hartline2008optimal} and \citep{goldner2024simple}.  \citet{hartline2008optimal} utilize the single-parameter nature of the problem, while the mechanism of \citet{goldner2024simple} relies heavily on the symmetry induced by the i.i.d. assumption, where the truthfulness is not maintained once the i.i.d. assumption is relaxed. 
Figure~\ref{fig:unit-demand} summarizes this hierarchy of the complexity of different variants of the problem.

\begin{figure}[h]  
    \centering
\begin{tikzpicture}
    \node[draw, rectangle, inner sep=5pt] (root) at (0,0) {$\substack{\text{\textbf{No-superstar}} \\ \Theta\left(\log(1 + \frac{n\cdot c}{m})\right) \\ \text{Theorem~\ref{thm:udB}}}$};
    \node[draw, rectangle, inner sep=5pt] (identical) at (-6,1.5) {$\substack{\text{\textbf{Identical items}} \\ \Theta\left(\log(1+\frac{n}{m})\right) \\  \text{\cite{hartline2008optimal}}}$};
    \node[draw, rectangle, inner sep=5pt] (iid) at (-6,-1.5) {$\substack{\text{\textbf{I.i.d. values}} \\ \Theta\left(\log(1+\frac{n}{m})\right) \\ \text{\cite{goldner2024simple}}}$};
    \node[draw, rectangle, inner sep=5pt] (adversarial) at (6,0) {$\substack{\text{\textbf{Adversarial}} \\ \Theta(\log n) \\ \text{Theorem~\ref{thm:ud}}}$};

    \draw[-] (identical.east) -- (root.west);
    \draw[-] (iid.east) -- (root.west);
    \draw[-] (root.east) -- (adversarial.west);
\end{tikzpicture}

\caption{Hierarchy of the complexity of different variants for unit-demand bidders.}
\label{fig:unit-demand}
\end{figure}


For the last instantiation, we consider the divisible goods case (Section~\ref{sec:divisible}), where we adapt the VCG with copies framework to the divisible goods setting by proposing the restricted capacity VCG mechanism. Rather than randomly selecting a number of copies of each item, the restricted capacity VCG mechanism randomly chooses a ``capacity'' for each item, i.e., a maximum amount of each item that each agent may receive. The mechanism then runs VCG on a version of the input modified by this capacity constraint. Observe that this mechanism is, in a sense, simpler than VCG with copies, since it circumvents the need for conversion of the allocations output by VCG since they are all feasible on the initial instance (since items are divisible).
In this case, we consider buyers with concave valuation functions\footnote{A valuation function $v_i$ is concave if for any two allocation vectors $\Vec{x}$ and $\Vec{y}$ and any $\lambda \in (0,1)$ we have that $\lambda v_i(\Vec{x}) + (1-\lambda) v_i(\Vec{y}) \leq v_i(\lambda \Vec{x} + (1-\lambda)\Vec{y})$.}, and show a restricted capacity VCG mechanism whose residual surplus is a $O(\log n)$ approximation to the social welfare (Theorem~\ref{thm:divisible}).  
We remark that this result is indeed tight. This is due to the lower bound on the performance of any randomized allocation mechanism for a single indivisible item shown in \cite{hartline2008optimal}.  Since fractional and randomized allocations coincide for the single-item case (when bidders are expected utility maximizers), we obtain a $\Omega(\log{n})$ lower bound on the performance of any mechanism in our multiple heterogeneous divisible goods setting.

In this paper, we also study the special case of two agents (Section~\ref{sec:two-agents}). 
First, we show that the case of a single-item is as hard as any number of items when agents have arbitrary monotone valuations; in both cases an approximation of $\frac{3}{2}$ is optimal. This result is consistent with our previously discussed results which show that the surplus guarantees achievable in several well-studied multi-parameter auction settings with many agents are no worse than the surplus guarantees achievable in the canonical single-item setting with many agents.
Then, in Section~\ref{sec:alt-benchmark}, we answer an open question that was raised by \citet{hartline2008optimal}, regarding the approximation ratio with respect to a different benchmark that is not the unattainable first-best.

\subsection{Additional Related Work}
In addition to \cite{hartline2008optimal}, the works perhaps closest to ours are \cite{fotakis2016efficient,goldner2024simple}.  \citet{fotakis2016efficient} study maximizing consumer surplus in fully general domains, but obtain only a $O(\log |\mathcal{O}|)$ approximation guarantee where $\mathcal{O}$ is the set of feasible outcomes.  Since $|\mathcal{O}|$ can be exponential in the number of items in our indivisible goods settings (and, essentially infinite in our divisible good settings), while their mechanisms apply in our settings, they guarantee only a trivial fraction of the social welfare as consumer surplus (which can be achieved simply by following \cite{hartline2008optimal} on the grand bundle of items).  
Seeking to improve upon the guarantees (in terms of bidders and items) of \cite{fotakis2016efficient}, \citet{goldner2024simple} initiated the study of \emph{sublinear} approximations in multi-dimensional settings.  Independently of our work, they studied multi-dimensional surplus maximization with unit-demand bidders in Bayesian settings as we do in Section \ref{sec:unit-demand}.  For the case where bidders' values are distributed i.i.d., they devise a \emph{Bayesian incentive compatible} (rather than truthful in expectation) auction.
In Section \ref{sec:unit-demand}, we show that our mechanism achieves the same asymptotic guarantee as \cite{goldner2024simple} under a weaker assumption where the probability of each item being the favorite of some agent is bounded by $\frac{c}{m}$ (which generalizes both the i.i.d. case considered in \cite{goldner2024simple} and the identical items unit-demand case considered in \cite{hartline2008optimal} when setting $c=1$).

Also quite similar to our work is that of \cite{qiao2023online,ganesh2023combinatorial}. \citet{qiao2023online} propose the problem of ``pen-testing''.  In the pen-testing problem, an algorithm faces a set of $n$ pens each with an unknown amount of ink.  The goal of the algorithm is to ``test'' these pens, possibly exhausting some of the ink contained within, in order to select a pen with a large amount of \emph{remaining} ink.  As \citet{qiao2023online} and \citet{ganesh2023combinatorial} argue, one can interpret this problem as a consumer surplus maximization problem where the pens are the bidders.  \citet{qiao2023online} then give algorithms for selecting a single pen in online pen-testing settings (where pens arrive one-by-one and the algorithm must make irrevocable decisions).  \citet{ganesh2023combinatorial} consider the problem of selecting multiple pens in both online and offline environments where a feasibility constraint restricts which subsets of pens can be simultaneously selected.  Notably, these works only apply in  \emph{binary-decision, single-parameter} mechanism design settings, since choosing or not choosing a pen corresponds to a binary service decision for the bidders (i.e., each bidder has a \emph{single} private value for receiving a service and either receives the service or not), whereas the problems we investigate are all multi-parameter auction settings (i.e., bidders hold multiple pieces of private information).

Similar in spirit to our question of mechanism design with money-burning is the literature on mechanism design \emph{without} money.  Mechanism design without money is a rich literature (see \cite{schummer2007mechanism} for a survey of some results) seeking to design mechanisms that align the incentives of the agents with the designer while avoiding monetary payments altogether. Indeed, one possible approach to the questions we explore in this work is to find a welfare maximizing mechanism without money.  However, it is known from prior work, e.g., \cite{hartline2008optimal}, that mechanisms without money cannot guarantee better than an $\Omega(n)$ approximation to the optimal social welfare, even in the case of allocating a single good.  As such, we seek mechanisms that leverage (small) monetary transfers in order to obtain better consumer surplus guarantees.

The objective in our work also has some similarities with the work on frugal mechanism design (see, e.g., \cite{archer2007frugal,chen2010frugal,talwar2003price,karlin2005beyond}), wherein an auctioneer seeks to acquire services from a set of strategic provider bidders and the auctioneer aims to minimize her total expenditure.  In both frameworks, the auctioneer aims to minimize the costs (i.e., the monetary transfers), but the algorithmic approaches must differ since frugal mechanism design concerns procurement (i.e., reverse) auctions whereas we study allocation (i.e., forward) auction problems.

Another line of work of \citet{guo2009worst}, \citet{guo2010optimal}, and \citet{moulin2009almost} aims to minimize the payments collected by the auctioneer in the VCG auction via redistributions, i.e., the payments collected by the auctioneer are then redistributed to the participating bidders.  Our mechanisms, by contrast, do not permit any transfer of money from the auctioneer to the participants (nor between the participants directly) and, thus, their approaches are insufficient for our purposes.  An interesting question is under what conditions one can obtain improved consumer surplus guarantees over the guarantees we achieve in this paper by allowing for money to be exchanged between bidders.

Another work that resembles our approach considers mechanisms with modified supplies of goods. In particular, \citet{roughgarden2012supply} design mechanisms that restrict the capacities of goods, whereas our approach takes the opposite route: we deliberately allocate additional goods and then apply rounding techniques to prevent over-allocation. While the former has proven useful for revenue maximization, we demonstrate that the latter is effective for maximizing surplus.

Finally, money burning sees practical application in the use of transaction fee mechanisms in blockchain (see, e.g., \cite{roughgarden2021transaction,chung2023foundations,chen2023bayesian,wu2023maximizing}).  Since in a blockchain protocol only a limited number of transactions can be posted to the chain per block, transaction fee mechanisms are used to decide which transactions are ultimately processed. 
In these mechanisms, some of the payments are burnt and are removed from the system altogether, rather than transferred to the miner who is processing the block.  In contrast to our problem, however, payments are burnt in transaction fee mechanisms in order to incentivize the miners (who are also assumed to be strategic agents in the model) to follow the blockchain protocol.  As such, the literature on transaction fee mechanism design centers, generally, only around designing mechanisms with good incentive properties (and can allow payments to be quite high), whereas in this work we aim to satisfy a simpler set of incentives while optimizing the more challenging objective of consumer surplus.

\section{Indivisible Goods}\label{sec:indivisible}
\subsection{Model}\label{sec:indivisble-model}
In this section, we consider settings where the mechanism designer wants to allocate a set $\items$ of  $m$ indivisible items to a set $\agents$ of $n$ agents. Each agent is associated with a valuation function $v_i:2^\items \rightarrow \reals_{\geq 0}$. We assume that valuation functions are monotone (i.e., for all $S\subseteq T \subseteq \items$, $v_i(S) \leq v_i(T)$) and normalized (i.e., $v_i(\emptyset)=0$).

The mechanism designer needs to design a truthful mechanism $\mech$ that receives the valuation functions of the agents $v_1,\ldots, v_n$ and returns an allocation  $A=(A_1,\ldots,A_n)\in \allocs$, where $\allocs$ is the set of all possible allocations of $\items$ and a payment vector  $ p=(p_1,\ldots,p_n) \in \reals_{\geq 0 }^n$, where the mechanism might use randomness, in which case, $A$ and $p$ are random variables returned by the mechanism. The mechanism designer's goal is to maximize the expected residual surplus, which is $\sum_{i\in\agents} v_i(A_i) - p_i$. 

We require that the designed mechanism satisfies truthfulness in expectation (TIE) and ex-post individual rationality (EPIR):
\begin{align*} &
\forall i,v_i,\tilde{v}_i , v_{-i}:~ \E_{(A,p) \sim \mech(v_i,v_{-i})}[v_i(A_i) - p_i] \geq 
\E_{(A,p) \sim \mech( \tilde{v}_i,v_{-i})}[v_i(A_i) - p_i]  & (TIE)
\\  &\forall i, v_i, , v_{-i}, (A,p) \in \support{\mech(v_1,\ldots,v_n)}:~ v_i(A_i) - p_i \geq 0 &
  (EPIR)
\end{align*}

We measure the performance of the mechanism using the first-best benchmark which is the maximal possible social welfare, which we denote by $\sw{\agents}{\items} = \max_{A\in \Delta} \sum_{i \in \agents}v_i(A_i)$. 
We say that a mechanism guarantees an $\alpha$-approximation (for $\alpha \geq 1$) of the optimal welfare as residual surplus if $$\forall{v_1,\ldots,v_n}:~ \E_{(A,p)\sim \mech(v_1,\ldots,v_n)} \left[\sum_{i\in\agents} v_i(A_i) - p_i\right] \geq \frac{\sw{\agents}{\items}}{\alpha}.$$

We consider several types of valuation functions:
\begin{itemize}
\item Valuation function $v_i:2^\items \rightarrow \reals_{\geq 0}$ is \textit{unit-demand} if for every non-empty set $S$, $v_i(S) = \max_{j \in S } v_i(\{j\})$.
\item Valuation function $v_i:2^\items \rightarrow \reals_{\geq 0}$ is \textit{gross-substitutes} if for every two vectors of item prices $p,p'\in \reals_{\geq 0}^m$, for which $p_j\leq p'_j$ for all $j \in \items$, and a set  $S \in \arg\max_{T \subseteq \items } v_i(T) - \sum_{j \in T} p(j)$, there exists a set $S' \in \arg\max_{T \subseteq \items} v_i(T) - \sum_{j \in T} p'(j)$ for which for all items $j\in S$ with $p_j=p'_j$ it holds that $j\in S'$.
\item Valuation function $v_i:2^\items \rightarrow \reals_{\geq 0}$ is \textit{submodular} if for every two sets $S,T\subseteq \items$ it holds that $v_i(S) +v_i(T) \geq v_i(S\cap T) + v_i(S \cup T)$.
\item Valuation function $v_i:2^\items \rightarrow \reals_{\geq 0}$ is a \textit{multi-unit} function if for every two sets $S,T\subseteq \items$ of the same size it holds that $v_i(S) =v_i(T)$.
\end{itemize}
A valuation function can belong simultaneously to multiple categories, and it is known that the class of unit-demand functions is a strict subset of the class of gross-substitutes functions which is a strict subset of the class of submodular functions \cite{LehmannLN06}. Among multi-unit functions the classes of gross-substitutes and submodular are the same \cite{LehmannLN06, GUL1999}.

\subsection{The VCG Mechanism with Copies}\label{sec:vcg-copies}
In this section, we introduce the VCG with copies framework. We then show how to use this framework and devise multiple truthful mechanisms with optimal guarantees for various settings. 
We begin by defining the key components of the framework.

\paragraph{Instance with Copies.} Consider an instance with a set $\items$ of items and a set $\agents$ of $n$ agents with valuation functions $v_1, \dots, v_n$. We define an instance with $2^\ell$ copies of each item in $\items$, yielding an item set $\items' =  \items \times [2^\ell]$. We define a projection function $g:2^{\items'} \to 2^\items$  that maps a subset of copied items back to their original items: $\items$, $g(S) = \{j \mid \exists k\in[2^\ell]: (j,k)\in S \}$. We define the \textit{extension} of $v:2^\items \rightarrow\reals_{\geq 0}$ to an instance with $2^\ell$ copies by   $v': 2^{\items'} \to \reals_{\geq 0}$ where $v'(S)=v(g(S))$.
Throughout the paper, we denote the valuations, sets, and payments of an instance with copies using a prime (i.e., $'$)—while the original instance’s valuations, sets, and payments are denoted without a prime.

\paragraph{VCG mechanism.}Our mechanisms use as a subroutine the VCG mechanism \cite{clarke1971, groves1973, vickrey1961}, which for a specific instance selects an arbitrary allocation that maximizes the social welfare, and each agent pays its negative externality to the other agents. That is, the allocation returned by the VCG mechanism $A$ is in $\arg\max_{A \in \Delta} \sw{\agents}{\items}$, and for all $i$, $p_i = \sw{\agents \setminus \{i\}}{\items} - \sw{\agents \setminus\{i\} }{\items \setminus A_i} $.

\paragraph{Rounding Scheme.}A key component of the mechanism we devise is a rounding scheme.
A rounding scheme takes as input $(A', \ell)$, where $A'$ is an allocation of an instance with $2^\ell$ copies of each item.
The rounding scheme (randomly) outputs an allocation $A$ for the original instance (without the copies).

An allocation $A'=(A'_1,\ldots,A_n')$ is \textit{non-redundant with respect to class of valuations $\mathcal{C}$} if there is no $S_i'$  which is a strict subset of $A'_i $ such that $v(S_i') = v(A'_i)$ for each $v\in \mathcal{C}$.

We say that a rounding scheme $B$ is \textit{$q$-valid} with respect to a class $\mathcal{C}$ of valuation functions if for every $\ell \geq 0$ and a non-redundant allocation $A'$ of $2^\ell$ copies of the items with respect to $\mathcal{C}$, $B$ allocates agent $i$  a (random) set $A_i$ such that 
for every $v\in \mathcal{C}$ it holds that 
\begin{equation}
\E[v(A_i)] = \frac{q}{2^\ell}\cdot v'(A'_i), \label{eq:q-valid}
\end{equation}
where $v'$ is the extension of $v$ to $2^\ell$ copies and the expectation is over the randomness of the rounding scheme.

The main challenge is to show how to devise a valid rounding scheme with constant $q$ for certain classes of valuations.
With these components in place, we can now formally define our mechanism.

\begin{algorithm}
\DontPrintSemicolon
\LinesNumbered
\KwIn{An item set $\items$, valuation profile $v_1,\ldots,v_n$,  parameter $r\in \naturals$, and a rounding scheme $B$}
\KwOut{An allocation $A=(A_1,\ldots,A_n)$ of $\items$ and a payment vector $p=(p_1,\ldots,p_n)$}
Draw $\ell$ uniformly at random from $\{0,\ldots,r\}$ \\
Let $\items' = \items \times [2^\ell]$ \tcp*{Create $2^\ell$ copies of each item} 
Let $v_i':2^{\items'} \rightarrow \reals_{\geq 0}$ be the extension of $v_i$ to $2^\ell$ copies \\
Calculate $(A',p') = \vcg{\items'}{v_1',\ldots,v_n'}$ \label{step:vcg} \\
Use the rounding scheme $B$ on $(A',\ell)$ to create an allocation $A$ \label{step:sub}\\
For every agent $i \in \agents$, set $p_i = p'_i\cdot\frac{v_i(A_i)}{v'_i(A'_i)}$ if $v'_i(A'_i) > 0$ and $p_i = 0$ otherwise. \label{step:payments}
\SetAlgoRefName{1}
\caption{VCG with copies}
\label{algo:vcg-copies}
\end{algorithm}
Mechanism~\ref{algo:vcg-copies} basically creates $2^\ell$ copies of all the items, and extends the valuation functions of the agents to the instance with the copies. Then, the mechanism calculates the VCG allocation and price vector for the instance with the copies.  
We assume that in Step~\ref{step:vcg}, the calculated allocation is non-redundant with respect to the class of valuations we consider.
Since in our actual instance, there is only a single copy of each item, we utilize a $q$-valid rounding scheme $B$ to generate an allocation for the original instance. The rounding scheme guarantees that the expected value that each agent receives is (in expectation) proportional to the value in the allocation returned by Step~\ref{step:vcg}. Lastly, in Step~\ref{step:payments}, we scale for each agent the payment returned by the VCG by the fraction of value allocated to this agent over the value returned by the VCG.

In the remainder of this section, we show that given a $q$-valid rounding scheme, our mechanism is TIE and EPIR (Claim~\ref{claim:epir-tie}).
We then devise $q$-valid rounding schemes for unit-demand valuations with $q = 1$ and for multi-unit submodular valuations with $q = 1/2$.
We then present residual surplus guarantees for the class of gross-substitutes valuations as a function of $q$ and $r$ (Lemma~\ref{lem:welfare-guarantee}).
Combining everything, we obtain an $O(\log(n))$-approximation.

The key challenges for establishing guarantees using our framework are: 1) Devising $q$-valid rounding schemes with large values of $q$ which we show is possible for unit-demand (or even additive subject to matroid constraints, see discussion in Section~\ref{sec:discussion}) and multi-unit submodular valuations. 2) Comparing the payments and utilities of the outcomes of the VCG mechanism for different amounts of copies. In particular, our key lemma is for the case of gross-substitutes valuations where we show that the additional contribution to the welfare by adding a copy of all items is at least half of the sum of payments of the VCG mechanism on the instance with the additional copies (Claim~\ref{cl:payment_bounded_by_SW_diff}).

We first claim that as long as the rounding scheme $B$ is $q$-valid  then  Mechanism~\ref{algo:vcg-copies} is truthful in expectation and ex-post individual rational no matter how the rounding scheme is implemented.
\begin{claim} \label{claim:epir-tie}
    If Mechanism~\ref{algo:vcg-copies} receives a $q$-valid rounding scheme $B$ with respect to a class $\mathcal{C}$ of valuation functions, then it is TIE and EPIR with respect to $\mathcal{C}$.
\end{claim}
\begin{proof}
    To show EPIR, we observe that since the VCG mechanism is EPIR we have $p'_i \leq v'_i(A'_i)$. Using this, if $v'_i(A'_i)$ is $0$, $i$ pays nothing.  Otherwise, the payment of agent $i$ for a set $A_i$ of items is $p_i=p'_i\cdot\frac{v_i(A_i)}{v'_i(A'_i)}\leq v_i(A_i)$, meaning that the agent pays no more than his value for the set of items.
    To show TIE, we fix a value of $\ell$. For a $q$-valid rounding scheme, the expected value of an agent $i$ is $\frac{q}{2^\ell}\times v_i(g(A'_i))$, where $A'_i$ is the allocation to agent $i$ in the VCG mechanism with $2^\ell$ copies and $g(A'_i)$ is the projection of those items to the actual instance. By Step~\ref{step:payments} and $B$ being a $q$-valid rounding scheme, the expected payment of agent $i$ is the same as in the VCG mechanism with $2^\ell$ copies of each item, that is $p'_i$.
    Observe that two statements imply that for every agent his expected utility is $\frac{q}{2^\ell}$ times his utility in the instance with $2^\ell$ copies.
    Since the VCG mechanism is truthful, reporting $v_i$ maximizes the agent's expected utility.
\end{proof}

We next show that for the cases of unit-demand and multi-unit valuations

\paragraph{Unit-demand valuations.}\label{unit-demand-subroutine} 
We construct a $1$-valid rounding scheme $B_{UD}$ for the class of unit-demand valuations, as follows:
For each item, independently select a random copy of it (among the $2^\ell$ copies). If this copy of the item is allocated to some agent, then allocate the item to this agent and charge accordingly. This is a $q$-valid rounding scheme since any non-redundant allocation returned at Step~\ref{step:vcg} with respect to the class of unit-demand valuations allocates at most one item per agent. An extension $v':2^{\items'}\rightarrow\reals_{\geq 0}$ of a unit demand function $v:2^\items\rightarrow \reals_{\geq 0} $ is unit-demand. Thus, the allocation $A'$ creates a partition of a subset of the agents according to the different items.
Therefore, allocating each item $j\in \items$ to each of agents a copy of it is allocated to with probability $2\ell$ independently is a $1$-valid rounding scheme for unit-demand valuations.

\paragraph{Multi-unit valuations.} We construct a $\frac{1}{2}$-valid rounding scheme $B_{MU}$ for the class of multi-unit valuations.
Given a non-redundant allocation $A'$ with respect to the class of multi-unit valuations, where each item has $2^\ell$ copies. We first create $n$ buckets of agents $B_1, \dots B_n$, where at the beginning each bucket contains a single agent. As long as there are two buckets, for which the number of items allocated to the agents within these buckets is at most $m$, we merge them into a single bucket. That is, for each bucket $B$, we maintain the condition: $|\bigcup_{i \in B} A'_i| \leq m$.
Once we stop merging, it must be that for every pair of remaining buckets, the total number of items allocated in every pair of buckets is at least $m+1$. I.e.,
 for any two different buckets $B,B'$, it holds that $|\bigcup_{i \in B\cup B'} A'_i| > m$
Since in the instance with the copies we have overall $m\cdot 2^{\ell}$ items, this implies that at the end of the process, at most $\left\lceil 2\frac{m\cdot 2^\ell}{m+1} \right\rceil \leq  2\cdot2^\ell$ buckets remain. 
Now, to implement a $\frac{1}{2}$-valid rounding scheme for multi-unit valuation functions we choose each bucket with probability $\frac{1}{2\cdot 2^\ell}$, and allocate only to agents that belong to this bucket the same amount of items as in $A'$. (If there are less than $2\cdot 2^\ell$ buckets, then with the remaining probability, we allocate nothing.) 
This allocation is feasible as in each bucket the sum of items allocated to the agents in this bucket is at most $m$. Moreover, each agent is allocated with probability $\frac{1}{2\cdot2^\ell}$, which satisfies Equation~\eqref{eq:q-valid}.

\begin{remark}
    We note that Mechanism~\ref{algo:vcg-copies} can be implemented in polynomial time for the two mentioned cases since both VCG and the subroutines can be implemented efficiently for the two cases. 
\end{remark}

\newcommand*{\xMin}{0}%
\newcommand*{\xMax}{6}%
\newcommand*{\yMin}{0}%
\newcommand*{\yMax}{5}%

\subsection{Analysis of the VCG Mechanism with Copies for Gross-Substitutes}\label{sec:vcg-analysis-gs}

We next present a welfare guarantee of our mechanism when applied to gross-substitutes valuation functions.  To do so, we abuse the notation and use $v_i$ instead of $v_i'$ for the instance with the copies. We can abuse this notation since for every set $S \subseteq \items'$, it holds that $v_i'(S) = v_i(g(S))$.  For a set of items $\items$, and a non-negative integer $c$, we denote by $c \cdot \items$ the superset that contains $c$ copies of each item of $\items$, and for two (super)sets $S_1, S_2,$ of items, we denote by $S_1 -  S_2$ the superset difference  where if an element appears $x_1$ times in $S_1$, and $x_2$ in $S_2$, it will appear  $\max(0,x_1-x_2)$ times in $S_1 -S_2$. For two supersets of items $S_1,S_2$, we denote by $S_1 + S_2$ the superset union of the two sets where if an element appears $x_1$ times in $S_1$ and $x_2$ times in $S_2$, then it will appear $x_1+x_2$ in $S_1+S_2$.
\begin{lemma}\label{lem:welfare-guarantee}
For a subclass $\mathcal{C}$ of gross-substitute valuation functions and a $q$-valid rounding scheme $B$ for the class $\mathcal{C}$, 
if Mechanism~\ref{algo:vcg-copies} receives valuations from $\mathcal{C}$, then the expected residual surplus of its output is at least  
$$\frac{q}{r+1}\left(\sw{\agents}{ \items}- \frac{\sw{\agents}{ 2^r\cdot \items}}{2^r}\right).$$ 
\end{lemma}
\begin{proof}
We first present a few useful properties of gross-substitutes valuations:
The first one is shown\footnote{They show a stronger claim, but the part in their claim regarding the modified cost function is equivalent to the following claim.} in \cite[Claims F.2.,~4.2,~F.1.,~F.3.,~F.4.,~F.5.]{BergerEFF23}:
\begin{claim}[\citep{BergerEFF23}]\label{cl:closure}
If $v_i$ is gross-substitutes (respectively, submodular, matroid rank function, coverage, XOS, subadditive), then so is $v_i'$.     
\end{claim}
The second property is shown in \cite{PAESLEME2017294}:
\begin{claim}[\cite{PAESLEME2017294}]\label{cl:convolution}
    Gross-substitutes functions are closed under convolution. I.e., for gross-substitutes $v_1,\ldots,v_n:2^\items \rightarrow \reals_{\geq 0}$, the function $v:\items\rightarrow \reals_{\geq 0}$ for which $v(S) = \max_{A\in \Delta(S)} \sum_{i} v_i(A_i)   $, where $\Delta(S)$ is the set of all allocations of the set of items $S$ to agent $1,\ldots,n$, then $v$ is also gross-substitutes. 
\end{claim}

We next define the following notation: For an agent $i\in \agents$,  and a set of items $\items'$ which consists of $k$ copies of the original set $\items$, let $p_i(\agents, \items')$ be the payment\footnote{We require that the payments are consistent with the implementation used in Step~\ref{step:vcg} of Mechanism~\ref{algo:vcg-copies}.} charged to agent $i$ in the VCG mechanism when applied to $\agents$ and $\items'$. We also denote by $p(\agents,\items') = \sum_{i\in \agents} p_i(\agents,\items')$.
Then, we first prove the following claim:

\begin{claim} \label{cl:payment_bounded_by_SW_diff}
    For agents with gross-substitutes valuations it holds that    $SW(\agents,2 \cdot \bar{\items})-SW(\agents, \bar{\items}) \geq \frac{p(\agents, 2 \cdot \bar{\items})}{2}$.
\end{claim}
\begin{proof}
    Let $\ali{i}{\agents}{2\cdot \bar{\items}}$ be the set allocated to agent $i$ by the VCG mechanism with the set of agents $\agents$ and $2$ copies of each item in $\bar{\items}$. 
    We have that
    \begin{align*}
        p_i(\agents, 2 \cdot \bar{\items}) & = \sw{\agents \setminus \{i\}}{ 2\cdot \bar{\items}}-\sw{\agents \setminus \{i\}}{ 2\cdot \bar{\items} - \ali{i}{\agents}{2\cdot \bar{\items}}} \\
        & = \sw{\agents\setminus \{i\}}{ 2\cdot\bar{\items}}-\sw{\agents}{ 2 \cdot\bar{\items}} +  v_{i}(\ali{i}{\agents}{2\cdot\bar{\items}}) \\
        & \leq \sw{\agents}{ 2\cdot \bar{\items} + \ali{i}{\agents}{2\cdot\bar{\items}}} - \sw{\agents}{ 2\cdot\bar{\items}} \\ 
         & \leq \sum_{j \in \ali{i}{\agents}{2\cdot\bar{\items}}} \left(\sw{\agents}{ 2\cdot\bar{\items} +\{j\}} - \sw{\agents}{ 2\cdot\bar{\items}} \right) ,
    \end{align*}
where the first equality is since the VCG payment of agent $i$ is equal to the decrease in the social welfare to the other agents by removing the set allocated to him, the second equality is due to the definition of $\ali{i}{\agents}{2\cdot\bar{\items}}$, the first inequality is due to the optimality of the social welfare of the VCG allocation, and the second inequality is since the social welfare function is gross-substitutes by Claim~\ref{cl:convolution} when applied on gross substitutes functions, and the valuation functions defined on the instance with the copies are gross substitutes by Claim~\ref{cl:closure}, and therefore it is submodular. 
    By summing over all $i$ we get that 
    \begin{equation}
        p(\agents,2\cdot \bar{\items}) = \sum_i  p_i(\agents, 2\cdot\bar{\items} ) \leq 2\sum_{j \in \bar{\items} } \left(\sw{\agents}{ 2\cdot\bar{\items} + \{j\}} - \sw{\agents}{ 2\cdot\bar{\items}} \right) \label{eq:p}
        \end{equation}

    On the other hand, if we order the elements in $\bar{\items}$ by $1,\ldots,\bar{m}$, then it holds that
    \begin{eqnarray}
         \sw{\agents}{ 2\cdot\bar{\items}} - \sw{\agents}{ \bar{\items}} 
        & = & \sum_{j=1}^{\bar{m}} \left(\sw{\agents}{ \bar{\items} +[j]} - \sw{\agents}{ \bar{\items} + [j-1]}\right) \nonumber \\
        & \geq & \sum_{j=1}^{\bar{m}} \left(\sw{\agents}{ 2\cdot \bar{\items}+ \{j\}} - \sw{\agents}{ 2\cdot \bar{\items}} \right) \nonumber
        \\ 
        & \geq  &\frac{ p(\agents, 2\cdot  \bar{\items})}{2}, 
        \nonumber\label{eq:gap-sw}
    \end{eqnarray}
where the first equality is by telescoping sum, the first inequality is again due to the social welfare function being submodular, and the second inequality is by Inequality~\eqref{eq:p}. This concludes the proof. 
\end{proof}
We are now ready to prove Lemma~\ref{lem:welfare-guarantee}.
For the convenience of notation, for the rest of the proof of the lemma, let  $\sws{i}$ be the social welfare of the instance with $i$ copies of each item, let $\sws{ 2^{-1}} = 0$, and let $p(i)$ be the sum of payments of the agents in the instance with $i$ copies. Then,  the welfare of Mechanism~\ref{algo:vcg-copies} is:

\begin{align}
     \E\left[\sum_{i\in\agents} v_i(A_i) \right] 
   &= \frac{q}{r+1}\cdot \sum_{l=0}^r  \frac{\sws{ 2^l}}{2^l} \nonumber\\
    & =  \frac{q}{r+1}\cdot\sum_{l=0}^r \sum_{i=0}^l \frac{\sws{ 2^i}- \sws{ 2^{i-1}}}{2^l}    \nonumber \\
    & =   \frac{q}{r+1}\cdot \sum_{i=0}^r (\sws{ 2^i}- \sws{ 2^{i-1}})\cdot 2\left( \frac{1}{2^i}-\frac{1}{2^{r+1}}\right) \nonumber \\ 
    & =    \frac{q}{r+1}\cdot \left(\sws{1}\cdot \left(2 - \frac{1}{2^r}\right) +\sum_{i=1}^r (\sws{ 2^i}- \sws{ 2^{i-1}})\cdot 2\left( \frac{1}{2^i}-\frac{1}{2^{r+1}}\right)\right), \label{eq:sw1} 
\end{align}
where the second equality is by telescoping sum; the third equality is by geometric sum; and the   last equality is since $SW(2^{-1}) = 0$.
On the other hand, the expected payment of the mechanism is:
\begin{eqnarray} \label{eq:sw2}
    \E\left[\sum_{i\in\agents} p_i\right]  =  \frac{q}{r+1}\cdot \sum_{l=0}^r \frac{p( 2^l)}{2^l} \leq  \frac{q}{r+1}\cdot \left(\sws{1}+\sum_{l=1}^r \frac{\sws{2^l} - \sws{2^{l-1}}}{2^{l-1}} \right),
\end{eqnarray} 
where the inequality is by Claim~\ref{cl:payment_bounded_by_SW_diff} and since $p(1) \leq \sws{1}$.
By combining Inequalities~\eqref{eq:sw1} and \eqref{eq:sw2} we get that the expected residual surplus is at least:

\begin{eqnarray*}
\E\left[\sum_{i\in\agents} v_i(A_i) - p_i \right]           & \geq &  \frac{q}{r+1}\cdot \left(\sws{1}\cdot \left(1 - \frac{1}{2^r}\right) -\sum_{i=1}^r \frac{\sws{ 2^i}- \sws{ 2^{i-1}}}{2^r} \right) \\
     &=& \frac{q}{r+1}\cdot \left(\sws{ 1}- \frac{\sws{ 2^r}}{2^r}\right),
\end{eqnarray*}
which concludes the proof.
\end{proof}

We get as immediate corollaries of Lemma~\ref{lem:welfare-guarantee} the following guarantees of  Mechanism~\ref{algo:vcg-copies}:
\begin{theorem}\label{thm:ud}
When Mechanism~\ref{algo:vcg-copies} (implemented with rounding scheme $B_{UD}$) receives unit-demand valuations, the expected residual surplus of its output is at least $O\left(\log(n)\right)$ approximation to the social welfare.
\end{theorem}
\begin{theorem}\label{thm:mu}
When Mechanism~\ref{algo:vcg-copies} (implemented with rounding scheme $B_{MU}$) receives submodular multi-unit valuations, the expected residual surplus of its output is at least $O\left(\log(n)\right)$ approximation to the social welfare.
\end{theorem}

We remark that the lower bound of \citet{hartline2008optimal} immediately implies a lower bound of $\Omega(\log(n))$ for both unit-demand bidders and multi-unit submodular bidders, making the results of Theorem~\ref{thm:ud} and Theorem~\ref{thm:mu} asymptotically tight.

\begin{proof}[Proof of Theorems~\ref{thm:ud} and \ref{thm:mu}]
By setting $r=2\log(n)$ and $q\geq 1/2$, we get that the guarantee of Lemma~\ref{lem:welfare-guarantee} satisfies
\begin{eqnarray*}
    \frac{q}{r+1}\left(\sw{\agents}{ \items}- \frac{\sw{\agents}{ 2^r\cdot \items}}{2^r}\right)   & \geq  & \frac{q}{r+1}\left(\sw{\agents}{ \items}- \frac{\sum_{i\in \agents } v_i(\items)}{2^r}\right) \\  
    &\geq &   \frac{1}{4\log(n)+2}\left(\sw{\agents}{ \items}- \frac{n\cdot \sw{\agents}{ \items}}{2^{2\log(n)}}\right) \\ & = &\Omega\left(\frac{\sw{\agents}{ \items}}{\log(n)}\right),
\end{eqnarray*}
where the first inequality is since no agent can get more than his value for the entire set, and the second inequality is by plugging $r$ and $q$ and since allocating all the items to one agent is a valid allocation and therefore $v_i(\items) \leq \sw{\agents}{\items}$.
\end{proof}

\newcommand{\MostValuableAgentsMap}{\mu^{-1}_{v}}
\newcommand{\MostValuableAgents}[1]{\mu^{-1}_{v}(#1)}
\newcommand{\MostValuableItemMap}{\mu_{v}}
\newcommand{\MostValuableItem}[1]{\mu_{v}(#1)}
\newcommand{\copiesmoreagents}{\log\frac{n\cdot (\frac{2e}{e-1})}{m}}
\newcommand{\copiesmoreagentsTwo}{\frac{n\cdot (\frac{2e}{e-1})}{m}}
\newcommand{\copiesmoreitems}{(\frac{2e}{e-1})}
\newcommand{\sharedconst}{\frac{e}{e-1}}

\newcommand{\AllocRule}{a}

\subsection{Bayesian Instances }\label{sec:unit-demand}

In this section, we show an improved analysis of the performance of  Mechanism~\ref{algo:vcg-copies} under Bayesian settings. 
In particular, we consider a setting where the valuation $v_i:2^{\items} \rightarrow \reals_{\geq 0}$ of each agent $i\in \agents$ is drawn from an underlying distribution $F_i$ satisfying the following assumption which we call the \textit{\assumption}assumption:
For each agent $i$, the probability that an item $j$ is agent $i$'s most valuable item is at most $\frac{c}{m}$ for some constant $c\geq 1$.\footnote{If there are multiple items that have the highest value, then the probability is counted only for one of the items. Our mechanism does not use this assumption and can be implemented regardless; we only use this assumption for the surplus guarantee for which it is sufficient that there exists a tie-breaking among the items that satisfies the \assumption assumption.}  I.e., for all $j\in\items$ it holds that $\Pr[v_i(
\{j\}) = \max_{j'\in\items} v_i(\{j'\})] \leq \frac{c}{m}$.
We also assume that the agents' valuations are independent, and we denote by $\mathcal{F} = \mathcal{F}_1 \times \dots \times \mathcal{F}_n$ the product of the agents' distributions. In Bayesian settings, we are interested in maximizing the expected performance of a mechanism, and so we say that a mechanism $M$ guarantees an $\alpha$-approximation (for $\alpha \geq 1$) of the optimal welfare as residual surplus if: $$\E_{v \sim \mathcal{F}}\left[\E_{(A,p)\sim \mech(v_1,\ldots,v_n)} \left[\sum_{i\in\agents} v_i(A_i) - p_i\right]\right] \geq \frac{\E_{v \sim \mathcal{F}}[\sw{\agents}{\items}]}{\alpha}.$$

Our result in this section is that it is possible to extract 
a better fraction of the social welfare as residual surplus (Theorem~\ref{thm:udB}). The case of $c=1$ of this result generalizes the case of identical items studied in \cite{hartline2008optimal}, 
and the case of unit-demand agents where each agent's values are distributed identically and independently from the same distribution, which was studied in \cite{goldner2024simple} independently from our work.

\begin{theorem}\label{thm:udB}
    For a set $\items$ of $m$ heterogeneous items, $n$ unit-demand agents $\agents$ with valuations sampled from the distribution $\mathcal{F}$ that satisfies the \assumption assumption for some constant $c\geq 1$, there exists a mechanism whose expected residual surplus is an $O(\log(1+\frac{n\cdot c}{m}))$  approximation to the welfare as residual surplus.
\end{theorem}
To prove this theorem, we first introduce some notation. 
Let $\MostValuableItemMap:\agents \to \items$ be a mapping from an agent to his most valuable item with respect to a valuation profile $v$, i.e.,   $\MostValuableItem{i} = \arg\max_{j\in \items} v_i(\{j\})$. If an agent $i$ has more than one item with maximal value, we define $\MostValuableItem{i}$ to be one of these items in a way that satisfies that $\Pr_{v_i \sim \mathcal{F}_i}[\MostValuableItem{i}= j]\leq \frac{c}{m}$, for every item $j \in \items$. Moreover, we assume that agent $i$ breaks ties consistently in an independent way from $v_{-i}$.
Similarly, let $\MostValuableAgentsMap:\items \to 2^\agents$ be a mapping from an item to the set of agents whose most valuable item is this item with respect to a valuation profile $v$, i.e., $\MostValuableAgents{j} = \{ i \mid \MostValuableItem{i}=j\}$.



The proof for Theorem \ref{thm:udB} heavily relies on the following lemma:

\begin{lemma}\label{shared-bound}
     For a set $\items$ of $m$ heterogeneous items, $n$ unit-demand agents $\agents$ with valuations sampled from the distribution $\mathcal{F}$ that satisfies the \assumption assumption for some constant $c\geq 1$, it holds that $$\E_{v\sim\mathcal{F}}[\sw{\agents}{\items}] \geq \frac{\E_{v\sim\mathcal{F}}[\sum_{i \in \agents } v_i(\{\MostValuableItem{i}\})]}{\frac{e}{e -1}\cdot \max\{{1,\frac{n\cdot c}{m}\}}}.$$
\end{lemma}


We first show how to derive Theorem~\ref{thm:udB} from Lemma~\ref{shared-bound}, and then we prove Lemma~\ref{shared-bound}.
\begin{proof}[Proof of Theorem~\ref{thm:udB}]
    We prove this theorem by considering Mechanism~\ref{algo:vcg-copies} with parameters $r= \lceil \log_2(\frac{2e}{e-1}\cdot   \max\{1, \frac{n\cdot c}{m}\})\rceil$, and rounding scheme $B_{UD}$ with $q=1$ (as defined in \nameref{unit-demand-subroutine}). Now, the expected residual surplus of this mechanism is:

\begin{eqnarray*}
& & \E_{v \sim \mathcal{F}}[\E_{(A,p) \sim \mech(v)}[\sum_{i \in \agents} v_i(A_i) - p_i]] \\ & \geq  & \frac{1}{r+1}\left(\E_{v\sim\mathcal{F}}[\sw{\agents}{ \items}]- \frac{\E_{v\sim\mathcal{F}}[\sw{\agents}{ 2^r\cdot \items}]}{2^r}\right)\\
& \geq & \frac{1}{r+1}\left(\E_{v\sim\mathcal{F}}[\sw{\agents}{ \items}]- \frac{\E_{v\sim\mathcal{F}}[\sum_{i \in \agents } v_i(\{\MostValuableItem{i}\})]}{2^r}\right) \\
& \geq & \frac{1}{r+1}\left(\E_{v\sim\mathcal{F}}[\sw{\agents}{ \items}]- \frac{(\sharedconst \cdot \max\{1, \frac{n\cdot c}{m}\})\cdot \E_{v\sim\mathcal{F}}[\sw{\agents}{ \items}]}{2^r}\right) \\
& \geq  & \frac{1}{r+1}\cdot\frac{\E_{v\sim\mathcal{F}}[\sw{\agents}{ \items}]}{2},
 \end{eqnarray*}
where the first inequality is due to Lemma~\ref{lem:welfare-guarantee}, the second inequality is since the social welfare of unit-demand agents is always bounded by giving each agent his favorite item, the third inequality is due to Lemma~\ref{shared-bound}, and the last inequality is by assigning the value of $r$. 

The theorem holds since $r \in O(\log(1+\frac{n\cdot c}{m}))$.
\end{proof}

Next, we prove Lemma~\ref{shared-bound}.
\begin{proof}[Proof of Lemma~\ref{shared-bound}]
    Consider the allocation rule $\AllocRule$: For a valuation profile $v$, $\AllocRule$ allocates each item $j$ to one of the agents in $\MostValuableAgents{j}$ uniformly at random.    
    To prove this lemma we bound the expected social welfare by the expected welfare of the allocation rule $\AllocRule$.
    The expected welfare of an agent $i$ under allocation rule $\AllocRule$ is equal to $\E_{v \sim \mathcal{F}}\left[\frac{v_i(\{\MostValuableItem{i}\})}{|\MostValuableAgents{\MostValuableItem{i}}|}\right]$. We next analyze the random variable of $|\MostValuableAgents{\MostValuableItem{i}}|$. We know that by definition $i \in\MostValuableAgents{\MostValuableItem{i}}$, and 
    by the \assumption assumption we have that for every $
    i'\in \agents \setminus \{i\}$, it holds that $\Pr[i'\in \MostValuableAgents{\MostValuableItem{i}} ] \leq \frac{c}{m}$. Combining with the independent assumption, we get that  $|\MostValuableAgents{\MostValuableItem{i}}|$ is   
      $1+ \sum_{i'\neq i} \mathbb{I}[\mu_v(i')=\mu_v(i)]$  which by our independence and \assumption assumptions is stochastically dominated by $1+ Bin(n-1, \frac{c}{m})$. Next we analyze the expected value of  $\frac{1}{\mid \MostValuableAgents{\MostValuableItem{i}}\mid}$.

        \begin{align}
        \E\left[\frac{1}{|\MostValuableAgents{\MostValuableItem{i}}|}\right]& \geq \E\left[\frac{1}{1+Bin(n-1,\frac{c}{m})}\right] \nonumber\\ 
        &=\sum_{k=0}^{n-1}\frac{1}{k+1}\binom{n-1}{k}\left(\frac{c}{m}\right)^k\cdot\left(1-\frac{c}{m}\right)^{n-1-k}  \nonumber\\
        &= \frac{m}{n\cdot c}\sum_{k=0}^{n-1}\binom{n}{k+1}\left(\frac{c}{m}\right)^{k+1}\cdot\left(1-\frac{c}{m}\right)^{n-(k+1)} \nonumber\\ 
        & = \frac{m}{n\cdot c}\cdot\left(1-\left(1-\frac{c}{m}\right)^{n}\right),  \label{expected_value_pr_getting_item}
    \end{align}
    where the second equality is due to the identity $\binom{n}{k} = \frac{n}{k}\binom{n-1}{k-1}$. The last equality follows by observing that $\sum_{k=0}^{n-1}\binom{n}{k+1}(\frac{c}{m})^{k+1}\cdot(1-\frac{c}{m})^{n-(k+1)}$ is equal to the probability that a binomial random variable with parameters $n$ and $\frac{c}{m}$  equals to some $k$ between $1$ and $n$, i.e., it equals to the probability that a binomial random variable with parameters $n$ and $\frac{c}{m}$ is not equal to $0$.

We note that since Inequality~\eqref{expected_value_pr_getting_item} holds for every realization of $\mu(i)$,  it implies that:
\begin{equation}
    \E_{v \sim \mathcal{F}}\left[\frac{v_i(\{\MostValuableItem{i}\})}{|\MostValuableAgents{\MostValuableItem{i}}|}\right] \geq \E_{v \sim \mathcal{F}}[v_i(\{\MostValuableItem{i}\})]\cdot \frac{m}{n\cdot c}\cdot\left(1-\left(1-\frac{c}{m}\right)^{n}\right). \nonumber 
\end{equation}
    Thus, 
    we get that the expected welfare of the allocation rule $\AllocRule$ is: 
 $$\E_{v\sim\mathcal{F}}\left[\sum_{i \in \agents } v_i(\{\MostValuableItem{i}\})\right]\cdot  \frac{m}{n\cdot c}\cdot\left(1-\left(1-\frac{c}{m}\right)^{n}\right).$$
    
    Now, we split the analysis depending on whether $n \geq m$, or $ n<m$.

    \paragraph{Case 1: $n\cdot c\geq m$.} In this case $ \frac{m}{n \cdot c}\cdot(1-(1-\frac{c}{m})^{n}) \geq \frac{m}{n\cdot c}\cdot(1-(1-\frac{c}{m})^{\frac{m}{c}}) \geq \frac{m}{n\cdot c}\cdot(1-\frac{1}{e})$, and so we get:
     $$
    \E_{v\sim\mathcal{F}}[\sw{\agents}{\items}]
    \geq  \E_{v\sim\mathcal{F}}\left[\sum_{i \in \agents } v_i(\{\MostValuableItem{i}\})\right]\cdot  \frac{m}{n\cdot c} \cdot \left(1-\frac{1}{e}\right),
    $$
    which concludes the proof of this case.
    \paragraph{Case 2: $n\cdot c < m$.} In this case $ \frac{m}{n\cdot c}\cdot(1-(1-\frac{c}{m})^{n})\geq  \frac{m}{n \cdot c}\cdot \left(1-{e}^{-\frac{cn}{m}}\right)$. Let $f(x)= x\cdot(1-e^{-\frac{1}{x}})$, and observe that this is an increasing function for $x > 0$. Since $n\cdot c<m$, we have that $$ \frac{m}{n\cdot c}\cdot\left(1-e^{\frac{-n\cdot c}{m}}\right) = f\left(\frac{m}{n\cdot c}\right) \geq f(1) = 1-\frac{1}{e} .$$ Overall for both cases:
    $$
    \E_{v\sim\mathcal{F}}[\sw{\agents}{\items}] \geq \E_{v\sim\mathcal{F}}\left[\sum_{i \in \agents } v_i(\{\MostValuableItem{i}\})\right]\cdot  \frac{m}{n\cdot c}\cdot\left(1-\left(1-\frac{c}{m}\right)^{n}\right) \geq \frac{\E_{v\sim\mathcal{F}}\left[\sum_{i \in \agents } v_i(\{\MostValuableItem{i}\})\right]}{\frac{e}{e-1}\cdot \max\{1, \frac{n\cdot c}{m}\}} .
    $$
    This concludes the proof of the lemma.    
\end{proof}

\section{Divisible Goods}\label{sec:divisible}
We now turn toward allocations with \emph{divisible goods}.  Since agents can be allocated fractions of each good, $\allocs$ is now the set of allocations which gives each agent $i$ a fraction $x_{ij} \in [0,1]$ of each good $j$ with the constraint that $\sum_{i \in \agents}{x_{ij}} \leq 1, \forall j \in \items$.  As before, each agent is associated with a valuation function $\vali : [0,1]^{|\items|} \rightarrow \mathbb{R}_{\geq 0}$ and we assume that valuation functions are monotone (i.e., $v_i(\Vec{x}) \leq v_i(\Vec{y})$ for all $\Vec{y} \geq \Vec{x}$ where comparison between two vectors is coordinate-wise) and normalized (i.e., $v_i(\Vec{0}) = 0$). 
In this section, we consider the case where agents' valuation functions are \emph{concave}.  We say that a valuation function $v_i$ is concave if for any two allocation vectors $\Vec{x}$ and $\Vec{y}$ and any $\lambda \in [0,1]$ we have that $\lambda v_i(\Vec{x}) + (1-\lambda) v_i(\Vec{y}) \leq v_i(\lambda \Vec{x} + (1-\lambda)\Vec{y})$.  Note that for any normalized and concave valuation function $v_i$ we have that  $v_i(\lambda \Vec{x}) \geq \lambda v_i(\Vec{x})$ for any $\lambda \in [0,1]$ and any allocation $\Vec{x}$.

We adapt the VCG with copies framework introduced in Section \ref{sec:vcg-copies} to the divisible items setting by proposing the restricted capacity VCG mechanism below.  We then use this mechanism to devise a truthful mechanism with optimal guarantees for the divisible goods setting we study.  Rather than randomly selecting a number of copies of each item, the restricted capacity VCG mechanism randomly chooses a ``capacity'' $q$ for each item, i.e., a maximum amount  $q$ of each item that each agent may receive.  The mechanism then runs VCG on a version of the input modified by this capacity constraint.  This mechanism is, in a sense, simpler than Mechanism \ref{algo:vcg-copies} since it circumvents the need for rounding any of the allocations output by VCG since they are all feasible on the initial instance (since items are divisible). As such, in this setting we are able to find a \emph{universally truthful} mechanism, i.e., a mechanism which is a randomization over truthful mechanisms, which provides a stronger incentive guarantee than truthfulness in expectation.

\begin{algorithm}
\DontPrintSemicolon
\LinesNumbered
\KwIn{An item set $\items$, valuation profile $v_1,\ldots,v_n$,  parameters $r\in \naturals$}
\KwOut{A fractional allocation $\mathbf{x}=(\Vec{x}_1,\ldots,\Vec{x}_n)$ of $\items$ and a payment vector $p=(p_1,\ldots,p_n)$}
Draw $\ell$ uniformly at random from $\{0,\ldots,r\}$\\
Let $q = 2^{-\ell}$\\
Let $v_i^q:[0,1]^{|\items|} \rightarrow \reals_{\geq 0}$ be the function $v_i^q(\Vec{x}) = v_i(\min\{q\Vec{1},\Vec{x}\})$ where the minimum is coordinate-wise\\
Calculate $(A,p) = \vcg{\items}{v_1^q,\ldots,v_n^q}$ \label{step:vcg-restricted}
\SetAlgoRefName{2}
\caption{Restricted capacity VCG}
\label{algo:vcg-restricted}
\end{algorithm}

We first show that Mechanism \ref{algo:vcg-restricted} is universally truthful and individually rational.  Essentially both of these facts follow from the properties of the VCG auction.
\begin{theorem}
    Mechanism \ref{algo:vcg-restricted} is universally truthful and ex-post individually rational.
\end{theorem}
\begin{proof}
    Observe that in line \ref{step:vcg-restricted} of Mechanism \ref{algo:vcg-restricted} one runs the VCG auction on a modified input valuations $\Vec{v}^q$.  This is equivalent to running a maximal-in-range mechanism on the original input valuations $\Vec{v}$ but only over outcomes that allocate at most a $q$ fraction of any good to any agent.  Maximal-in-range mechanisms are known to be truthful and ex-post individually rational (see, e.g., \cite{nisan2007computationally}).  But then, Mechanism \ref{algo:vcg-restricted} is a randomization over truthful and ex-post individually rational mechanisms and, hence, Mechanism \ref{algo:vcg-restricted} is universally truthful and ex-post individually rational.
\end{proof}

Before we prove our main theorem of this section demonstrating the approximation guarantee of our auction, we give some useful facts about concave valuation functions over divisible goods.

\begin{claim}\label{cl:concave-min}
Let $v_i : [0,1]^{|\items|} \rightarrow \reals_{\geq 0}$ be a concave function. Then for every $q\geq 0$, the function $v_i^q : [0,1]^{|\items|} \rightarrow \reals_{\geq 0}$ defined as $v_i^q(\Vec x)  = v_i(\min\{q\Vec 1, \Vec x\})$ is concave.
\end{claim}
\begin{proof}
    Consider an arbitrary triple $(x,y,q) \in [0,1]^3$ and an arbitrary $\lambda \in (0,1)$.  We will demonstrate that \[\lambda\min\{x,q\} + (1-\lambda)\min\{y,q\} \leq \min\{q, \lambda x + (1-\lambda)y\}.\]  If $x \leq q$ and $y \leq q$ then the left-hand side is $\lambda x + (1-\lambda) y$ and the right-hand side is as well since $\lambda \in (0,1)$.  If $x > q$ and $y > q$ then both the left-hand and right-hand sides equal $q$.  Finally, if $x > q$ and $y \leq q$ (the remaining case is symmetric) then the left-hand side is $\lambda q + (1-\lambda)y$.  Since  $y \leq q$ we have that $\lambda q + (1-\lambda)y \leq q$.  On the other hand, since $q < x$ we have that $\lambda q + (1-\lambda)y < \lambda x + (1-\lambda)y$.  Thus, we have that the left-hand side is less than or equal to the right-hand side in this case as well.

    With this fact in hand, we may conclude the proof by observing that 
    \begin{align*}
        \lambda v_i^q(\Vec x) + (1-\lambda)v_i^q(\Vec y) &= \lambda v_i(\min\{\Vec x, q\Vec 1\}) + (1-\lambda)v_i(\min \{\Vec y, q \Vec 1\})\\
        &\leq v_i(\lambda\min\{\Vec x, q\Vec 1\} + (1-\lambda)\min\{\Vec y, q \Vec 1\})\\
        &\leq v_i(\min\{q\Vec 1, \lambda \Vec x + (1-\lambda) \Vec y\})\\
        &=v_i^q(\lambda \Vec x + (1-\lambda) \Vec y)
    \end{align*}
    where the equalities follow by the definition of $v_i^q$, the first inequality follows from the concavity of $v_i$, and the second inequality follows from applying our above inequality coordinate-wise on the vectors $\Vec x, \Vec y,$ and $q \Vec 1$. 
\end{proof}

\begin{claim}\label{cl:concave-closure}\cite{phelps2009convex}
Concave functions are closed under convolution.  I.e., for concave $v_1, \dots, v_n : [0,1]^{|\items|} \rightarrow \reals_{\geq 0}$, the function $v : \items \rightarrow \reals_{\geq 0}$ for which $v(S) = \max_{\Vec x \in \Delta(S)}\sum_{i}{v_i(\Vec{x}_i})$, where $\Delta(S)$ is the set of feasible allocations of the set of items $S$ to agents $1, \dots, n$, then $v$ is also concave.
\end{claim}

As in Section \ref{sec:indivisible}, a critical portion of our main proof is that the payments are bounded by the difference in social welfare function for two different, but related inputs.  We denote by $SW(\agents, q, M)$ the optimal social welfare (i.e., the social welfare of the VCG mechanism) to $\agents$ (i.e., the set of all agents) with allocation capacity $q$ and set of items (with corresponding sizes) $M$. We also denote by $p_i(\agents, q, M)$ the payment of agent $i$ when the VCG mechanism is applied to $\agents$ with item set $M$ and with allocation capacity $q$ and let $p(\agents, q, M) = \sum_{i \in \agents}{p_i(\agents, q, M)}$.  Analogously to Claim \ref{cl:payment_bounded_by_SW_diff} for the indivisible goods case, we obtain Claim \ref{cl:payment_bound_divisible} below (whose proof we defer to Appendix \ref{sec:omitted-proofs} due to space constraints).

\begin{claim}\label{cl:payment_bound_divisible}
    For agents with concave valuations it holds that $SW(\agents,q,\items)-\frac{1}{2}SW(\agents, 2q,\items) \geq \frac{p(\agents, q,\items)}{2}$.
\end{claim}

With Claim \ref{cl:payment_bound_divisible} in hand, we are ready to prove the main theorem of this section, which demonstrates the approximation ratio of Mechanism \ref{algo:vcg-restricted}.

\begin{theorem} \label{thm:divisible}
When Mechanism~\ref{algo:vcg-restricted} receives concave valuations it obtains expected residual surplus that is a $O\left(\log(n)\right)$ approximation to the social welfare.
\end{theorem}

We again note that the lower bound of \citet{hartline2008optimal} implies a lower bound of $\Omega(\log(n))$ for buyers with concave valuation functions (see Section~\ref{sec:contribution}), therefore the result of Theorem~\ref{thm:divisible} asymptotically tight.

\begin{proof}
    For the rest of the proof, we hold the set of agents and items fixed and are concerned only with the allocation restriction $q$.  As such, for ease of presentation, in the remainder of this section, we let $SW(q)$ denote $SW(\agents, q, \items)$ and $p(q)$ denote $p(\agents, q, \items)$.
    
    Then, using the convention that $SW(2^1) = 0$, we get that the residual surplus times $\log{n} + 1$ is:
    \begin{align*}
        (\log{n} + 1)\cdot\mathbb{E}\left[\sum_{i \in \agents}v_i(x_i) - p_i\right]
        &=\sum_{\ell = 0}^{\log n}{\left(SW(2^{-\ell}) - p(2^{-\ell})\right)} \\
        &= \sum_{\ell = 0}^{\log n}{\sum_{t = 0}^{\ell}{2^{-\ell+t}SW(2^{-t})- 2^{-\ell + (t -1)}SW(2^{-(t-1)})}} - \sum_{\ell = 0}^{\log{n}}p(2^{-\ell})\\
        &= \sum_{t = 0}^{\log n}{\sum_{\ell = t}^{\log n}{2^{-\ell+t}SW(2^{-t})- 2^{-\ell + (t -1)}SW(2^{-(t-1)})}} - \sum_{\ell = 0}^{\log{n}}p(2^{-\ell})\\   
        &= \sum_{t = 0}^{\log n}{\left(SW(2^{-t}) - \frac{1}{2}SW(2^{-(t-1)})\right)\cdot\left(2-\frac{2^{t}}{n}\right)} - \sum_{t = 0}^{\log{n}}{p(2^{-t})}.\\
    \end{align*}
    The first equality above applies the definition of residual surplus, the second applies a telescoping sum, the third interchanges the order of summation, and the final equality is a geometric sum.

    We may rewrite the first summation in the above expression as:
    \begin{align}
        &\sum_{t = 0}^{\log n}{\left(SW( 2^{-t}) - \frac{1}{2}SW(2^{-(t-1)})\right)\cdot\left(2-\frac{2^{t}}{n}\right)}\nonumber\\
        &= SW(2^0)\cdot \left(2-\frac{1}{n}\right) +\sum_{t = 1}^{\log n}{\left(SW(2^{-t}) - \frac{1}{2}SW(2^{-(t-1)})\right)\cdot\left(2-\frac{2^{t}}{n}\right)}\nonumber\\
        &\geq SW(2^0)\cdot\left(2 - \frac{1}{n}\right) +\sum_{t = 1}^{\log n - 1}{\left(SW(2^{-t}) - \frac{1}{2}SW(2^{-(t-1)})\right)\cdot\left(2-\frac{2^{t}}{n}\right)}\label{eq:sw-bound-divis}, 
    \end{align}
    where the equality applies the fact that we have set $SW(2^1)$ to be $0$ and the inequality uses the fact that giving each agent half of what they receive in the social welfare maximizing allocation with capacity $2q$ is a feasible allocation when capacities are $q$ (and hence $SW(q) \geq 1/2 \cdot SW(2q)$ for every $q\leq \frac{1}{2}$).

    On the other hand, we may bound the payments (i.e., the second summation in the expression for the residual surplus) as:
    \begin{align}
        \sum_{t = 0}^{\log{n}}{p(2^{-t})} &= \sum_{t = 0}^{\log{n}-1}{p(2^{-t})}        \leq p(2^0) + \sum_{t = 1}^{\log n - 1}{2\cdot\left(SW(2^{-t}) - \frac{1}{2}SW(2^{-(t-1)})\right)}\label{eq:pay-bound-divis},
    \end{align}
    where the equality uses the fact that if the allocation capacity of each item is capped at $1/n$ then the payments are $0$ (since each agent has no externality) and the inequality applies Claim \ref{cl:payment_bound_divisible} to each term in the summation with $t \geq 1$.
    Finally, we may combine Inequalities \eqref{eq:sw-bound-divis} and \eqref{eq:pay-bound-divis} to obtain
    \begin{align*}
        &(\log{n} + 1)\cdot\mathbb{E}\left[\sum_{i \in \agents}v_i(x_i) - p_i\right]\\
        &\geq SW(2^0)\cdot\left(2 - \frac{1}{n}\right) - p(2^0) -\frac{1}{n}\sum_{t=1}^{\log{n}-1}{2^t\left(SW(2^{-t}) - \frac{1}{2}SW(2^{-(t-1)})\right)}\\
        &\geq SW(2^0)\cdot\left(1-\frac{1}{n}\right)-\frac{1}{n}\sum_{t=1}^{\log{n}-1}{2^t\left(SW(2^{-t}) - \frac{1}{2}SW(2^{-(t-1)})\right)}\\
        &= SW(2^0) - \frac{1}{n} \sum_{t=0}^{\log{n}-1}{\left(2^tSW(2^{-t}) - 2^{t-1}SW(2^{-(t-1)})\right)}\\
        &=SW(2^0) - \frac{1}{n}\cdot\frac{n}{2}SW\left(\frac{2}{n}\right)\\
        &\geq \frac{1}{2}SW(2^0),
    \end{align*}
    where the second inequality follows from the fact that in the VCG mechanism the social welfare is more than the sum of the payments, the second equality is the computation of the telescoping sum, and the final inequality uses the fact that $SW(2^0) \geq SW(2/n)$.
\end{proof}


\section{Discussion and Future Directions} \label{sec:discussion}
In this paper, we study mechanism-design for maximizing consumer surplus.
Our work is the first to achieve a sub-linear approximation to the welfare in adversarial multi-parameter settings, 
and we obtain asymptotically optimal guarantees for various settings including an $O(\log( n))$-approximation for unit-demand and multi-unit submodular agents,
an $O(\log{(1 + \frac{n\cdot c}{m})})$-approximation for Bayesian unit-demand settings satisfying the no-superstar-item assumption,
and an $O(\log(n))$-approximation for agents with concave valuations over divisible goods.
We also resolve some additional open questions regarding instances with two bidders from \cite{hartline2008optimal}.

We remark that all our mechanisms are prior-independent (or  prior-free) while the matching lower bounds are achieved in instances where the agents' values are independently sampled from the same distribution and the mechanisms are only required to be Bayesian incentive compatible.
A natural future direction is whether this is true for other settings as well, whether Bayesian incentive compatible mechanisms have the same worst-case guarantees as prior-independent ones. 




\paragraph{Beyond Unit-Demand and Multi-Unit.} We next discuss how to adapt Mechanism~\ref{algo:vcg-copies} for a broader class of valuations.
Consider the case where agents' valuations are weighted matroid valuation functions under known matroids. In other words, each agent $i\in\agents $ is associated with a (publicly known to the auctioneer) matroid $M_i=(\items,E_i)$, and private values  $v^i_1,\ldots,v^i_m \in \reals_{\geq 0} $ where agent's $i$ valuation function $v_i:2^\items \rightarrow \reals_{\geq 0} $ satisfies $v_i(S) = \max_{T\in E_i} \sum_{j\in S\cap T} v^i_j$.
An allocation $A' = (A'_1, \dots A'_n)$ is non-redundant with respect to the class of weighted matroid functions (for known matroids) if $g(A'_i) \in E_i$ and an agent is not allocated two copied of the same item.
For this case, we construct a $1$-valid rounding scheme that allocates each item independently with probability $\frac{1}{2^\ell}$ to one of its recipients of the copies. This is indeed a $1$-valid rounding scheme since matroids are downward-closed, which implies that the expected value each agent receives by independent allocation satisfies Equation~\eqref{eq:q-valid}.  

Mechanism ~\ref{algo:vcg-copies} applied with this rounding scheme is TIE and ex-post IR by Claim~\ref{claim:epir-tie}, and provides an $O(\log n)$ approximation by Lemma~\ref{lem:welfare-guarantee} as weighted matroid functions are gross-substitutes, and since weighted matroid functions are closed\footnote{By Claim~\ref{cl:closure} under the transformation that creates copies, and by similar arguments the same holds with respect to weighted matroid functions.} under the operation of adding copies of items. For achieving the $O(\log n)$-approximation, the mechanism needs to know the underlying matroids $M_i$ (if the matroids are private information and the agents report them, then the truthfulness proof requires that the rounding scheme should be applied to all possible allocations and not just allocations that allocate feasible sets to the agents).


We leave as an open direction to extend our framework (or variants of it) beyond the settings studied in this paper, or in general devise truthful mechanisms with optimal surplus guarantees for other classes of functions (e.g., submodular, XOS, subadditive, supermodular, and additive subject to downward-closed constraints).

\paragraph{An alternative Benchmark}
In this paper, we considered mechanism design for surplus maximization and compared it to the first-best benchmark, i.e., the social welfare. An interesting open question is establishing guarantees with respect to tighter benchmarks for surplus maximization. One option is trying to generalize the $\mathcal G$ benchmark introduced in \cite{hartline2008optimal} for single-parameter settings. However, the definition of $\mathcal G$ relies heavily on the single-parameter structure, and it is not clear whether there is an intuitive way to generalize it to multi-dimensional settings. Furthermore, as discussed in \cite{goldner2024simple}, finding a tractable alternative benchmark seems to be a difficult problem in its own right. They demonstrate this by examining two prominent approaches in revenue maximization - bounding the optimal revenue via Langrangian duality \cite{duality-paper} or by generating copies of the agents \cite{Chawla-copies} - and highlighting the challenges in applying these methods to multi-dimensional surplus maximization settings.

An open direction for future research could involve investigating additional concepts of truthfulness, such as universal truthfulness (studied briefly in Section \ref{sec:divisible}) and Bayesian incentive compatibility. While the primary focus of this paper has been on mechanisms that are truthful in expectation, these alternate notions present valuable areas for further exploration.

\bibliographystyle{plainnat}
\bibliography{bibliography}
\appendix
\section{Omitted Proofs}\label{sec:omitted-proofs}


\subsection{Proof of Claim \ref{cl:payment_bound_divisible}}
\begin{proof}[Proof of Claim \ref{cl:payment_bound_divisible}]
    Our proof proceeds similarly to the proof of Claim \ref{cl:payment_bounded_by_SW_diff}.  Let $A_i(\agents,q,M)$ denote the bundle of goods allocated to agent $i$ by the VCG mechanism ran with valuation vector $\Vec{v}^q$ (equivalently, with valuation vector $\Vec v$ but allocation capacities $q$) and items $M$.  We have that
    \begin{align*}
        p_i(\agents, q, \items) &= SW(\agents \setminus \{i\}, q,\items) - SW(\agents, q,\items) + v_i^q(A_i(\agents,q,\items)) \\
        &= SW(\agents \setminus \{i\}, q,\items) - SW(\agents, q,\items) + v_i(A_i(\agents,q,\items)) \\
        &\leq SW(\agents,q,\items+A_i(\agents,q,\items)) - SW(\agents,q,\items),
    \end{align*}
    where the first equality applies the VCG payment of agent $i$, the second applies the fact that each entry in $A_i^q$ is at most $q$ since it is a bundle allocated by the social-welfare maximizing allocation with restriction parameter $q$, and the inequality is due to the optimality of the social welfare of the VCG allocation.
    Summing the above inequality over agents then gives 
    \begin{equation}\label{eq:price-divis-lhs}
    \sum_{i \in \agents}{p_i(\agents,q,\items)} \leq \sum_{i \in \agents}{SW(\agents,q,\items+A_i(\agents,q,\items)) - SW(\agents,q,\items)}.
    \end{equation}
    
    Index the agents in an arbitrary order from $1$ to $n$ and let $B_i(\agents, q, M)$ denote the bundle containing all the goods allocated from agents $1$ to $i$, inclusive, in the VCG allocation with capacity $q$ and item set $M$, i.e., $B_i(\agents, q, M) = \sum_{k = 1}^{i}{A_i(\agents,q,M)}$ where the sum is taken coordinate-wise on the allocation vectors.  Then, we also have that
    \begin{align}
        SW(\agents,q,\items) - \frac{1}{2}SW(\agents,2q,\items) &\geq SW(\agents,q,\items) - SW\left(\agents,q,\frac{\items}{2}\right)\nonumber \\
        &=\sum_{i \in \agents}{SW\left(\agents,q,\frac{\items}{2} + B_i(\agents, q, \frac{\items}{2})\right) - SW\left(n,q,\frac{\items}{2}-B_{i-1}(\agents, q, \frac{\items}{2})\right)}\nonumber\\
        &\geq \sum_{i \in \agents}{SW\left(\agents,q,\items+A_i(\agents,q,\frac{\items}{2})\right) - SW(\agents,q,\items)}\nonumber\\
        &\geq \frac{1}{2}\sum_{i \in \agents}{SW(\agents,q,\items + A_i(\agents,q,\items)) - SW(\agents,q,\items)}, \label{eq:price-divis-rhs}
    \end{align}
    where the first inequality is due to the fact that the resulting allocation from giving each agent exactly half of what she receives in $SW(\agents,2q,\items)$ is a feasible allocation when the item input is $\items/2$ and capacity is $q$, the first equality is by telescoping sum, the second inequality is due to the fact that by Claims \ref{cl:concave-min} and \ref{cl:concave-closure} we have that the social welfare function is concave, and the final inequality is also due to this fact.
    Combining Inequalities \eqref{eq:price-divis-lhs} and \eqref{eq:price-divis-rhs} then completes the proof of the claim.
\end{proof}

\section{Improved Approximations for Two Agents}\label{sec:two-agents}
In light of our results in the previous sections, which show that the surplus guarantees achievable in several well-studied multi-parameter auction settings with many agents are no worse than the surplus guarantees achievable in the canonical single-item setting with many agents, a natural question is whether there exist settings which are more challenging from an approximation standpoint than the single-item setting.  To explore this question further, we now turn to the interesting special case of two agents.  In previous work of \cite{hartline2008optimal}, it was demonstrated that there exists a mechanism for two agents and a single item that achieves surplus which is a $3/2$-approximation to the optimal social welfare\footnote{More precisely, they demonstrate that their mechanism achieves a $3/2$-approximation to a weaker benchmark $\mathcal{G}$ which we discuss in Subsection \ref{sec:alt-benchmark}.  The proof of this guarantee, however, implies a $3/2$-approximation to the stronger benchmark of optimal social welfare, which is of interest to the majority of our paper.}.
We generalize this result with Mechanism \ref{algo:two-agents-many-items} which, as we demonstrate in Theorem \ref{thm:two-agents-welfare}, obtains a $3/2$-approximation guarantee in the case of two agents with \emph{arbitrary} monotone valuation functions over an item set $\items$.

\begin{algorithm}
\DontPrintSemicolon
\KwIn{An item set $\items$, a valuation profile of two agents $v_1,v_2$}
\KwOut{An allocation $A=(A_1,A_2)$ of $\items$ and a payment vector $p=(p_1,p_2)$}
With probability $1/3$: Run the VCG auction on the reported valuations \\
With probability $1/3$: Allocate the grand bundle of items to bidder $1$ and charge nothing \\
With remaining $1/3$ probability: Allocate the grand bundle of items to bidder $2$ and charge nothing
\SetAlgoRefName{3}
\caption{A two bidder, multi-item auction with optimal approximation for monotone valuations}
\label{algo:two-agents-many-items}
\end{algorithm}

\begin{theorem}\label{thm:two-agents-welfare}
    Mechanism \ref{algo:two-agents-many-items} obtains surplus utility that is a $3/2$-approximation to the optimal social welfare.
\end{theorem}
\begin{proof}
    Let $A^*_i$ denote the bundle allocated to bidder $i$ in the optimal solution.  We may then write the total social welfare of the optimal solution is $v_1(A^*_1) + v_2(A^*_2)$.  Observe that the surplus obtained by running the VCG mechanism on the entire reported valuations is \begin{align*}
        \text{Surp}(\texttt{VCG})&=v_1(A^*_1) + v_2(A^*_2) - ((v_2(\items) - v_2(A^*_2)+(v_1(\items) - v_1(A^*_1))) \\
        &= 2(v_1(A^*_1) + v_2(A^*_2)) - (v_1(\items) + v_2(\items)).
    \end{align*}
    On the other hand, we have that the surplus obtained by allocating the grand bundle uniformly at random to one of the two agents is
    \begin{align*}
        \text{Surp}(\texttt{RA}) &= \frac{v_1(\items) + v_2(\items)}{2}.
    \end{align*}
    But then, by running the VCG mechanism with probability $1/3$ and the random allocation mechanism with probability $2/3$ we obtain that the resulting surplus of Mechanism \ref{algo:two-agents-many-items} is 
    \begin{align*}
        \frac{1}{3}\cdot(2(v_1(A^*_1) + v_2(A^*_2)) - (v_1(\items) + v_2(\items))) + \frac{2}{3}\cdot \frac{v_1(\items) + v_2(\items)}{2} = \frac{2}{3}\cdot (v_1(A^*_1) + v_2(A^*_2)),
    \end{align*}
    as desired.
\end{proof}

We now complement this result by providing an instance based on a construction of \cite{hartline2008optimal} involving two agents and a single item which shows that $3/2$ is the best approximation possible, even in the case of a single-item.  These two results together then imply that, perhaps surprisingly, for the special case of two agents (disregarding computational and communication concerns) allocating a single-item to maximize surplus is as hard as allocating any number of items when agents have \emph{arbitrary} monotone valuations.

\begin{theorem}
    No mechanism for two agents can achieve surplus utility that is a $3/2-\epsilon$-approximation  for $\epsilon > 0$ to the optimal social welfare even in the case of a single-item and even when the agents draw their values independently from identical distributions known to the mechanism.
\end{theorem}
\begin{proof}
    Consider two agents drawing their values i.i.d. from a standard exponential distribution.  As demonstrated by \cite{hartline2008optimal}, the mechanism that maximizes the surplus for such a setting randomly allocates the item to one of the two agents and the expected surplus of such a mechanism is $1$.  On the other hand, the expected social welfare is equal to the expectation of the maximum of two random variables $X_1$ and $X_2$ drawn independently from an exponential distribution.  We may compute this as 
    \begin{align*}
    \int_0^\infty\text{Pr}[\max\{X_1,X_2\} > z]~dz &= \int_0^\infty\text{Pr}[X_1 > z] +\text{Pr}[X_2 > z] - \text{Pr}[X_1 > z \wedge X_2 > z] ~dz \\
    &= \int_0^\infty\text{Pr}[X_1 > z] +\text{Pr}[X_2 > z] - \text{Pr}[X_1 > z]\cdot\text{Pr}[X_2 > z] ~dz\\
    &= \int_0^\infty e^{-z} + e^{-z} - e^{-2z}~dz =3/2.\qedhere
    \end{align*}
\end{proof}

\subsection{Alternative Benchmarks}\label{sec:alt-benchmark}
Recall that the mechanisms in Sections \ref{sec:indivisible} and \ref{sec:divisible} achieve surplus guarantees which are optimal approximations to the social welfare.  However, these guarantees are asymptotic and logarithmic, raising the question of whether or not there exist tighter benchmarks that can better separate the performance of the two mechanisms.  As all prior-free mechanisms are anonymous (i.e., cannot treat differently agents which are, a priori, identical), we would like to compare the performance of a prior-free mechanism to the surplus achieved by the best anonymous, truthful mechanism.  This observation led to the development of the benchmark $\mathcal{G}(\Vec v)$ in \cite{hartline2008optimal}.  Formally, $\mathcal{G}(\Vec v) = \sup_{F} \text{Opt}_F(\Vec v)$ where $\text{Opt}_F(\Vec v)$ is the surplus of the Bayesian optimal mechanism when all bidders draw their value from a known distribution $F$ and the realization of values is $\Vec v$.  \citet{hartline2008optimal} provide a prior-free mechanism which obtains a $O(1)$-approximation to $\mathcal{G}$ in the case of unit-demand bidders and \emph{identical} items.  In addition, they raise an open question of determining the best-possible approximation ratio when the benchmark is $\mathcal{G}$ ``even in the two agent, one unit special case'' \cite{hartline2008optimal}.  We resolve this open question below.

\begin{algorithm}
\DontPrintSemicolon
\KwIn{A single item, two values $v_1$ and $v_2$}
\KwOut{An allocation of the item and a payment}
Re-index the bidders and values in non-increasing order of the reported value (i.e., such that $v_1 \geq v_2$) \\
\If{$v_1 > 3v_2$}{With probability $4/5$ allocate the item to bidder $1$ and charge bidder $1$ a price of $\frac{5v_2}{4}$\\With the remaining probability allocate the item to bidder $2$ and charge bidder $2$ a price of $0$}
\If{$v_1 \leq 3v_2$}{With probability $1/2$ allocate the item to bidder $1$ and charge bidder $1$ a price of $\frac{v_2}{5}$\\With the remaining probability allocate the item to bidder $2$ and charge bidder $2$ a price of $\frac{v_1}{5}$}
\SetAlgoRefName{4}
\caption{A two bidder, single-item auction with optimal approximation to $\mathcal{G}$}
\label{algo:G-optimal}
\end{algorithm}

\begin{theorem}\label{thm:G-truthful}
    Mechanism \ref{algo:G-optimal} is truthful-in-expectation, ex-post individually rational, and obtains surplus utility that is a $5/4$-approximation to $\mathcal{G}(\Vec v)$.
\end{theorem}
\begin{proof}
    Recall that in single-parameter environments, a mechanism is truthful if it has a monotone allocation rule and charges payments corresponding to the identity given in \cite{myerson1981optimal}.  By observation, it is clear that Mechanism \ref{algo:G-optimal} is monotone, i.e., the probability a bidder is allocated a good is (weakly) increasing with her bid (holding the other bid constant).  As such, it remains to verify that the prices charged by the mechanism are the threshold payments (i.e., those implied by the payment identity of \citet{myerson1981optimal}).  Consider varying the value $v_i$ of bidder $i$ when the value of the other bidder $j$ is fixed to $v_j$.  When $0 \leq v_i < \frac{1}{3}v_j$ the probability that bidder $i$ receives the good is $0.2$ and she pays $0$, as expected by the payment identity.  When $\frac{1}{3}v_j \leq v_i \leq 3v_j$ the probability that $i$ receives the good is $0.5$.  As such, by the payment identity, $i$ should be charged $\frac{1}{3}v_j \cdot (0.5 - 0.2) = \frac{1}{10}v_j$ and this is precisely the price the mechanism charges (normalized by the probability to be allocated which is half). 
    Finally, when $3v_j < v_i$ the probability that $i$ receives the good is $0.8$. But then, by the payment identity, $i$ should be charged $3v_j \cdot (0.8 - 0.5) + \frac{1}{10}v_j = v_j$ which is precisely what the mechanism charges (again, normalized by the allocation probability $\frac{4}{5}$).

    Now, we turn to the analysis of the mechanism's utility. Index the bidders in decreasing order of value such that we have  $v_1 \geq v_2$.  As demonstrated in \cite{hartline2008optimal}, in the special case of two agents and a single item, $\mathcal{G}(\Vec v)$ is equal to the better of performing a Vickrey auction or random allocation.  In other words, we have that $\mathcal{G}(\Vec v) =  \max\left\{v_1-v_2, \frac{v_1+v_2}{2}\right\}$\footnote{We note that in the proof of Proposition 4.1 in \cite{hartline2008optimal}, after arguing that $\mathcal{G}(\Vec v)$ is the better of a Vickrey auction or random allocation in the case of two agents, they incorrectly compute that $\mathcal{G}(\Vec v) = \max\left\{v_1-\frac{v_2}{2}, \frac{v_1+v_2}{2}\right\}$.  This leads them to give an incorrect lower bound of $4/3 > 5/4$ on the approximation achievable by any mechanism.} 

Consider first the case that $v_1 > 3v_2$.  Then, we have that $\frac{v_1}{2} > \frac{3v_2}{2}$ and thus $\mathcal{G}(\Vec v) = v_1 - v_2 > \frac{v_1 + v_2}{2}$. 
In this case, the mechanism obtains surplus \[\frac{4}{5}(v_1-\frac{5v_2}{4}) + \frac{1}{5}v_2 = \frac{4}{5}\mathcal{G}(\Vec v).\]  
On the other hand, if $v_1 \leq 3v_2$ we have that $\mathcal{G}(\Vec v) = \frac{v_1 + v_2}{2} \geq v_1 - v_2$.  In this case, our mechanism obtains surplus \[\frac{1}{2}(v_2 -\frac{v_1}{5})  + \frac{1}{2}(v_2-\frac{v_1}{5}) = \frac{4}{10}(v_1 + v_2) = \frac{4}{5}\mathcal{G}(\Vec v).\qedhere\]
\end{proof}

We complement the upper bound that Mechanism \ref{algo:G-optimal} achieves by demonstrating that it is the best possible with Theorem \ref{thm:lower-bound-G} below.

\begin{theorem}\label{thm:lower-bound-G}
    No mechanism for two agents can achieve surplus utility that is a $5/4-\epsilon$-approximation  for $\epsilon > 0$ to $\mathcal{G}(\Vec v)$ even in the case of a single-item and even when the agents draw their values independently from identical distributions known to the mechanism.
\end{theorem}
\begin{proof}
    Our proof proceeds very similarly to the proof of Proposition 4.1 in \cite{hartline2008optimal} (essentially, correcting an arithmetic error that propagates therein).  Consider two agents drawing their values i.i.d. from a standard exponential distribution.  Again, the mechanism which maximizes the expected surplus for such a setting randomly allocates the item to one of the two agents and the expected surplus of such a mechanism is $1$.  On the other hand, the expected value of $\mathcal{G}(\Vec v)$ is \[\mathbb{E}[\mathcal{G}(\Vec v)] = \mathbb{E}\left[\max\left\{v_1 - v_2, \frac{v_1 + v_2}{2}\right\}\right].\]  Observe that $v_1 - v_2 \leq \frac{v_1 + v_2}{2}$ if and only if $v_1 \leq 3v_2$.

    Conditioning on the smaller value $v_2$ we have that $v_1 = v_2 + x$ for $x \geq 0$ with $x$ exponentially distributed.  We then observe that
    \begin{align*}
        \mathbb{E}[\mathcal{G}(v_1,v_2) | v_2] &= \int_0^{2v_2}\left(v_2 + \frac{x}{2}\right)e^{-x}dx + \int_{2v_2}^\infty xe^{-x}dx=\frac{1}{2} \cdot \left(2v_2 + 1 + e^{-2v_2}\right).
    \end{align*}
    Finally, we know that the minimum of two standard exponentially distributed variables is distributed like an exponential distribution with parameter $2$ so we may integrate out to obtain
    \begin{align*}
        \mathbb{E}[\mathcal{G}(v_1,v_2)] &= \int_0^{\infty}{2e^{-2x}\cdot \frac{1}{2} \cdot \left(2x + 1 + e^{-2x}\right)dx} =\frac{5}{4}.\qedhere
    \end{align*}
\end{proof}

\end{document}